\documentclass{sig-alternate}

\usepackage[utf8]{inputenc}
\usepackage{url}
\usepackage{graphicx,bbm,stmaryrd}
\usepackage{lmodern}
\usepackage[paper=letterpaper,width=7in,height=9.25in,centering]{geometry}
\usepackage{microtype}

\conferenceinfo{ISSAC 2010,}{25--28 July 2010, Munich, Germany.}
\CopyrightYear{2010}
\crdata{978-1-4503-0150-3/10/0007}

\boilerplate={This work is in the public domain. As such, it is not
subject to copyright. Where it is not legally possible to consider this
work as released into the public domain, any entity is granted the right
to use this work for any purpose, without any conditions, unless such
conditions are required by law.}
\toappear{\the\boilerplate\par
{\confname{\the\conf}} \the\confinfo\par}

\newtheorem{proposition}{Proposition}
\newtheorem{theorem}{Theorem}
\newdef{example}{Example}
\newdef{remark}{Remark}

\newcommand{\fixdefsectionspacing}{\vspace*{-8pt}}

\begin{document}

\setlength{\emergencystretch}{1em}

\title{NumGfun: a Package for Numerical and Analytic Computation with \mbox{D-finite}
Functions}\author{Marc Mezzarobba\\
\affaddr{Algorithms Project-Team, INRIA Paris-Rocquencourt, France}\\
\email{marc.mezzarobba@inria.fr}}
\maketitle

\begin{abstract}
  This article describes the implementation in the software package
  NumGfun of classical algorithms that operate on solutions
  of linear differential equations or recurrence relations with
  polynomial coefficients, including what seems to be
  the first general implementation of the fast high-precision
  numerical evaluation algorithms of Chudnovsky~\& Chudnovsky. In some cases,
  our descriptions contain improvements over existing algorithms. We
  also provide references to relevant ideas not currently used in NumGfun.
\end{abstract}

\medskip\noindent\textbf{Categories and Subject Descriptors:} I.1.2 [\textbf{Symbolic and Algebraic Manipulation}]: Algorithms

\medskip\noindent\textbf{General Terms:} Algorithms, Experimentation, Theory


\section{Introduction}\label{sec:intro}

Support for computing with {\emph{\mbox{D-finite} functions}}, that is, solutions of
linear differential equations with polynomial coefficients, has become a
common feature of computer algebra systems. For instance, Mathematica now
provides a data structure called \texttt{DifferentialRoot} to represent
arbitrary \mbox{D-finite} functions by differential equations they satisfy and
initial values. Maple's \texttt{DESol} is similar but more limited.

An important source of such general \mbox{D-finite} functions is combinatorics, due
to the fact that many combinatorial structures have \mbox{D-finite} generating
functions. Moreover, powerful methods allow to get from a combinatorial
description of a class of objects to a system of differential equations that
``count'' these objects, and then to extract precise asymptotic information
from these equations, even when no explicit counting formula is
available~{\cite{FlajoletSedgewick2009,Stanley1999}}. A second major
application of \mbox{D-finiteness} is concerned with special functions. Indeed, many
classical functions of mathematical physics are \mbox{D-finite} (often by virtue of
being defined as ``interesting'' solutions of simple differential equations),
which allows to treat them uniformly in algorithms. This is exploited by the
\emph{Encyclopedia of Special
Functions}~{\cite{MeunierSalvy2003}}
and its successor under development, the \emph{Dynamic Dictionary of
Mathematical
Functions}~\cite{DDMF}, an
interactive computer-generated handbook of special functions.

These applications require at some point the ability to perform ``analytic''
computations with \mbox{D-finite} functions, starting with their numerical
evaluation. Relevant algorithms exist in the literature. In particular,
\mbox{D-finite} functions may be computed with an absolute error bounded by $2^{- n}$
in $n \log^{O (1)} n$ bit operations---that is, in softly linear time in
the size of the result written in fixed-point notation---at any point of
their Riemann surfaces~{\cite{ChudnovskyChudnovsky1990}}, the necessary error
bounds also being computed from the differential equation and initial
values~{\cite{vdH1999}}. However, these algorithms have remained
theoretical~{\cite[\S9.2.1]{Dupont2006}}.
The ability of
computer algebra systems to work with \mbox{D-finite} functions is (mostly)
limited to symbolic manipulations, and the above-mentioned fast evaluation
algorithm has served as a recipe to write numerical evaluation routines for
specific functions rather than as an algorithm for the entire class of
\mbox{D-finite} functions.

This article introduces NumGfun, a Maple package that attempts to fill this
gap, and contains, among other things, a general implementation of that
algorithm. NumGfun is distributed as a subpackage of
gfun~\cite{SalvyZimmermann1994},
under the GNU LGPL. Note that it comes with help pages: the goal of the
present article is not to take the place of user documentation, but rather to
describe the features and implementation of the package, with supporting
examples, while providing an overview of techniques relevant to the development of similar software. The following examples illustrate typical
uses of NumGfun, first to compute a remote term from a combinatorial sequence, then to evaluate a special function to high precision
near one of its singularities.

\begin{example}
  \label{ex:Motzkin}The Motzkin number $M_n$ is the number of ways of drawing
  non-intersecting chords between $n$ points placed on a circle. Motzkin
  numbers satisfy $(n + 4) M_{n + 2} = 3 (n+1) M_n + (2 n + 5) M_{n + 1}$. Using
  NumGfun, the command
  \texttt{nth\_term(\{(n+4)*M(n+2)\hspace{0pt plus 1ex}=\hspace{0pt plus 1ex}3*(n+1)*M(n)+(2*n+5)*M(n+1), M(0)=1,M(1)=1\},M(n),\emph{k})} computes
  $M_{10^5} = 6187\ldots7713 \simeq 10^{47\,706}$ in{\footnote{All
  timings reported in this article were obtained with the following
  configuration: Intel~T7250 CPU at 2~GHz, 1~GB of RAM, Linux~2.6.32,
  Maple~13, GMP~4.2.1.}} 4.7~s and $M_{10^6} = 2635 \ldots 9151 \simeq 
  10^{477\,112}$ in 1~min. 
  Naïvely unrolling the recurrence (using Maple) takes 10.7~s for
  $M_{10^5}$, and 41~s for $M_{2\cdot10^5}$.
  On this (non-generic) example, \texttt{nth\_term} could be made competitive for smaller indices by taking advantage of the fact that the divisions that occur while unrolling the recurrence are exact.
\end{example}

\begin{example}
  \label{ex:Heun}The double confluent Heun function $U_{\alpha, \beta,
  \gamma, \delta}$ satisfies $(z^2 - 1)^3 U''(z) + {(2 z^5 - 4 z^3 - \alpha
  z^4 + 2 z + \alpha)} U'(z) + {(\beta z^2 + (2 \alpha + \gamma) z +
  \delta)} U(z) =
  0$, $U (0) = 1$, $U' (0) = 0$.
  It is singular at $z=\pm 1$.
  The command
  \texttt{evaldiffeq(\emph{eq},y(z), -0.99,1000)}
  where \texttt{\emph{eq}} is this differential equation
  yields $U_{1, \frac{1}{3}, \frac{1}{2}, 3} (- 0.99) \approx 4.67755
  \text{...(990 digits)...} 05725$ in 22~s.
\end{example}
\fixdefsectionspacing

\paragraph*{Related work}Most algorithms implemented in NumGfun originate in
work of Chudnovsky~\& Chudnovsky and of van der Hoeven. Perhaps the most
central of these is the ``bit burst'' numerical evaluation
method~\cite{ChudnovskyChudnovsky1990}. It belongs to the
family of binary splitting algorithms for \mbox{D-finite} series, hints at the
existence of which go back to~{\cite[\S178]{HACKMEM}}, and generalizes earlier work of
Brent~{\cite{Brent1976}} for specific elementary functions. Better known (thanks
among others to~{\cite{HaiblePapanikolaou1997}}) binary
splitting algorithms can be seen as special cases of the bit burst algorithm.
One reason why, unlike these special cases, it was not
used in practice is that in~{\cite{ChudnovskyChudnovsky1990}}, none of the
error control needed to ensure the accuracy of the computed result is part of
the algorithm. Van der
Hoeven's version~{\cite{vdH1999}} addresses this issue, thus giving a
full-fledged evaluation algorithm for the class of \mbox{D-finite} functions, as
opposed to a method to compute any \mbox{D-finite} function given certain bounds.

These algorithms extend to the computation of limits of \mbox{D-finite}
functions at singularities of their defining equation. The case of regular
singularities is treated both
in~{\cite{ChudnovskyChudnovsky1986+1987,ChudnovskyChudnovsky1990}}, and more
completely in~{\cite{vdH2001}}, that of irregular singular
points in~{\cite{vdH2007}}. See~{\cite[\S12.7]{Bernstein2008}}, {\cite[\S1]{vdH2007}} for more history and context. 

On the implementation side, routines based on binary splitting for the
evaluation of various elementary and special functions are used in
general-purpose libraries such as CLN~{\cite{CLN}} and
MPFR~{\cite{FousseHanrotLefevrePelissierZimmermann2007,Jeandel2000}}. Binary
splitting of fast converging series is also the preferred algorithm of
software dedicated to the high-precision computation of mathematical constants
on standard PCs, including the current record holder for~$\pi$~{\cite{Bellard2010}}. Finally, even the impressive range of built-in
functions of computer algebra systems is not always sufficient for
applications. Works on the implementation of classes of ``less common'' special
functions that overlap those considered in NumGfun
include~{\cite{AbadGomezSesma2008,FalloonAbbottWang2003}}.

This work is based in part on the earlier~\cite{Mezzarobba2007}.

\paragraph*{Contribution}
The main contribution presented in this article is
NumGfun itself. We recall the algorithms it uses, and
discuss various implementation issues. Some
of these 
descriptions include
improvements or details that do not seem to have appeared elsewhere.
Specifically:
(i)~we give a new variant of the analytic continuation algorithm for D-finite functions that is faster with respect to the order and degree of the equation;
(ii)~we improve the complexity analysis of the bit burst algorithm by a factor $\log \log n$;
(iii)~we point out that Poole's method to construct solutions of differential equations at regular singular points can
be rephrased in a compact way in the language of noncommutative polynomials, leading to faster evaluation of \mbox{D-finite} functions in these points;
and
(iv)~we describe in some detail the practical computation of all the bounds needed to obtain a provably correct result.

\paragraph*{What NumGfun is not}Despite sharing some of the algorithms used to
compute mathematical constants to billions of digits, our code aims to
cover as much as possible of the class of \mbox{D-finite} functions, not to break
records. Also, it is limited to ``convergent'' methods: asymptotic expansions, summation to the least term, and resummation of divergent power series are currently out of scope.

\paragraph*{Terminology}Like the rest of gfun, NumGfun works with
{\emph{\mbox{D-finite} functions}} and {\emph{P-recursive sequences}}. We recall only
basic definitions here; see {\cite[\S6.4]{Stanley1999}} for further
properties. A formal power series $y \in \mathbbm{C}[[z]]$ is
{\emph{\mbox{D-finite}}} over $K \subseteq \mathbbm{C}$ if it solves a non-trivial
linear differential equation
\begin{equation}
  \label{eq:deq}
  y^{(r)} (z) + a_{r - 1} (z)\, y^{(r - 1)} (z) + \cdots + a_0 (z)\, y (z) = 0
\end{equation}
with coefficients $a_k \in K (z)$. The same definition applies to analytic
functions. A sequence $u \in \mathbbm{C}^{\mathbbm{N}}$ is
{\emph{P-recursive}} over~$K$ if it satisfies a non-trivial linear recurrence
relation
\begin{equation}
\label{eq:rec}
  u_{n + s} + b_{s - 1} (n)\, u_{n + s - 1} + \cdots + b_0 (n)\, u_n = 0,
  \; b_k \in K (n).
\end{equation}
A sequence $(u_n)_{n \in \mathbbm{N}}$ is P-recursive if and only if its
generating series $\sum_{n \in \mathbbm{N}} u_n z^n$ is \mbox{D-finite}.

The poles of the coefficients $a_k$ of~(\ref{eq:deq}) are its {\emph{singular
points}};  nonsingular points are called {\emph{ordinary}}. In gfun, a
\mbox{D-finite} function is represented by a differential equation of the
form~(\ref{eq:deq}) and initial values at the origin, which we assume to be an
ordinary point. Similarly, P-recursive sequences are encoded by a recurrence
relation plus initial values, as in Ex.~\ref{ex:Motzkin} above. If $y (z) =
\sum_{k = 0}^{\infty} y_k z^k$, we let $y_{; n} (z) = \sum_{k = 0}^{n - 1} y_k
z^k$ and $y_{n ;} (z) = \sum_{k = n}^{\infty} y_k z^k$.

The {\emph{height}} of an object is the maximum bit-size of the integers
appearing in its representation: the height of a rational number $p / q$ is
$\max (\lceil \log p \rceil, \lceil \log q \rceil)$, and that of a complex
number (we assume that elements of number fields $\mathbbm{Q}(\zeta)$ are
represented as $\sum_i x_i \zeta^i / d$ with $x_i, d \in \mathbbm{Z}$),
polynomial, matrix, or combination thereof with rational coefficients is the
maximum height of its coefficients. We assume that the bit complexity $M (n)$
of $n$-bit integer multiplication satisfies $M (n) = n (\log n)^{O (1)}$, $M
(n) = \Omega (n \log n)$, and $M (n + m) \geq M (n) + M (m)$, and that
$s\times s$ matrices can be multiplied in $O(s^\omega)$ 
operations in their coefficient ring.

\section{Binary splitting}\label{sec:binsplit}

``Unrolling'' a recurrence relation of the form~(\ref{eq:rec}) to compute
$u_0, \ldots, u_N$ takes $\Theta (N^2 M (\log N))$ bit operations, which is
almost linear in the total size of $u_0, \ldots, u_N$, but quadratic in that
of~$u_N$. The binary splitting algorithm computes a single term~$u_N$ in
essentially linear time, as follows:~(\ref{eq:rec}) is first reduced to a
matrix recurrence of the first order with a single common denominator:
\begin{equation}\label{eq:bsrec}
  q (n) U_{n + 1} = B (n) U_n, \hspace{1em} B (n) \in \mathbbm{Z}[n]^{s \times
  s}, q (n) \in \mathbbm{Z}[n], 
\end{equation}
so that $U_N = P (0, N) \, U_0 / \bigl(\prod_{i = 0}^{N - 1} q (i)\bigr)$, where $P
(j, i) = {B(j-1)} \cdots B(i+1) B(i)$. One then computes $P (0, N)$
recursively as $P (0, N) = P ( \left\lfloor N / 2 \right\rfloor, N) P (0,
\left\lfloor N / 2 \right\rfloor)$, and the denominator $\prod_{i = 0}^{N - 1} q (i)$ in a
similar fashion (but separately, in order to avoid expensive gcd
computations).

The idea of using product trees to make the most of fast multiplication dates
back at least to the seventies~\cite[\S12.7]{Bernstein2008}.
The general statement below is
from~{\cite[Theorem~2.2]{ChudnovskyChudnovsky1990}}, except that the authors
seem to have overlooked the cost of evaluating the
polynomials at the leaves of the tree. 

\begin{theorem}[Chudnovsky, Chudnovsky] \label{thm:binsplit}
  Let $u$ be a P-recursive sequence over $\mathbbm{Q}(i)$, defined
  by~\eqref{eq:rec}. Assume that the
  coefficients $b_k(n)$ of~\eqref{eq:rec} have no poles in~$\mathbbm{N}$.
  Let $d, h$ denote bounds on their degrees (of numerators and
  denominators) and heights, and $d', h'$ corresponding bounds for the
  coefficients  of $B (n)$ and $q (n)$ in~\eqref{eq:bsrec}.
  Assuming $N \gg s, d$, 
  the binary splitting algorithm outlined above computes
  one term~$u_N$ of $u$ in 
  $O (s^{\omega} M (N (h' + d' \log N)) \log (Nh'))$, 
  that is, 
  $O (s^{\omega} M (sdN (h + \log N)) \log (Nh))$, bit
  operations.
\end{theorem}

\begin{proof}[sketch]
  Write $H=h'+d' \log N$.
  Computing the product tree $P (0, N)$ takes $O (s^{\omega} M (NH) \log
  N)$ bit operations~\cite[\S2]{ChudnovskyChudnovsky1990} (see also
  Prop.~\ref{prop:cste} below), and
  the evaluation of each leaf $B (i)$ may be done in $O(M(H) \log d')$
  operations~{\cite[\S3.3]{BostanCluzeauSalvy2005}}.
  This gives $u_N$ as a fraction that is simplified in $O(M(NH) \log(NH))$
  operations~{\cite[\S1.6]{BrentZimmermann2009-MCAv0.4}}. 

  Now consider how~\eqref{eq:rec} is rewritten into~\eqref{eq:bsrec}. With
  coefficients in $\mathbbm{Z}[i]$ rather than $\mathbbm{Q}(i)$, the~$b_k(n)$
  have height $h'' \leq {(d + 1)} h$. To get $B(n)$ and $q(n)$, it remains
  to reduce to common denominator the whole equation; hence $d' \leq sd$
  and $h' \leq s {(h'' + \log s + \log d)}$.
  These two conversion steps take $O (M (sdh \log^2 d))$ and $O (M (d' h' \log s))$
  operations respectively, using product trees.
  The assumption $d, s = o(N)$ allows to write $H=O(sd(h+\log N))$ and
  get rid of some terms, so that the total complexity simplifies as stated.
\end{proof}

Since the height of $u_N$ may be as large as $\Omega ((N
+ h) \log N)$, this result is optimal with respect to $h$ and $N$, up to
logarithmic factors. The same algorithm works over any algebraic number field instead of $\mathbbm{Q}(i)$. This is useful for evaluating \mbox{D-finite} functions ``at
singularities'' (\S\ref{sec:regsing}). More generally, similar complexity
results hold for product tree computations in torsion-free
$\mathbbm{Z}$-algebras (or $\mathbbm{Q}$-algebras: we then write $A
=\mathbbm{Q} \otimes_{\mathbbm{Z}} A'$ for some $\mathbbm{Z}$-algebra $A'$ and
multiply in $\mathbbm{Z} \times A'$), keeping in mind that, without basis
choice, the height of an element is defined only up to some additive constant.

\paragraph*{Constant-factor improvements}
Several techniques permit to improve the constant
hidden in the $O (\mathinner\cdot)$ of Theorem~\ref{thm:binsplit}, by
making the computation at each node of the product tree less expensive. We
consider two models of computation.

In the {\emph{FFT model}}, we assume that the complexity of long
multiplication decomposes as $M (n) = 3 F (2 n) + O (n)$, where $F (n)$
is the cost of a discrete Fourier transform of size~$n$ (or of another related linear
transform, depending on the algorithm).
FFT-based integer multiplication algorithms adapt to reduce the multiplication of two matrices of height~$n$ in $\mathbbm{Z}^{s\times s}$ to $O(n)$ multiplications of matrices of height~$O(1)$, for a total of $O (s^2 M (n) + s^{\omega} n)$ bit operations. This is known as ``FFT addition''~{\cite{Bernstein2008}}, ``computing in the FFT
mode''~{\cite{vdH1999}}, or ``FFT invariance''.
A second improvement (``FFT doubling'', attributed to
R.~Kramer in~{\cite[\S12.8]{Bernstein2008}}) is specific to the computation of
product trees. The observation is that, at an internal node where operands of
size $n$ get multiplied using three FFTs of size~$2 n$, every second
coefficient of the two direct DFTs is already known from the level
below.

The second model is {\emph{black-box multiplication}}. There, we may use fast
multiplication formulae that trade large integer multiplications for additions
and multiplications by constants. The most obvious example is that the
products in $\mathbbm{Q}(i)$ may be done in four integer multiplications using
Karatsuba's formula instead of five with the naïve algorithm. In general,
elements of height $h$ of an algebraic number field of degree $d$ may be
multiplied in $2 dM (h) + O (h)$ bit operations using the Toom-Cook algorithm.
The same idea applies to the matrix multiplications. Most classical matrix
multiplication formulae such as Strassen's are so-called bilinear algorithms.
Since we are working over a commutative ring, we may use more general
{\emph{quadratic}}
algorithms~{\cite[\S14.1]{BuergisserClausenShokrollahi1997}}. In
particular, for all~$s$,
Waksman's algorithm~{\cite{Waksman1970}} multiplies $s \times s$ matrices over
a commutative ring~$R$ using $s^2 \lceil s / 2 \rceil + (2 s - 1) \lfloor s / 2
\rfloor$ multiplications in~$R$, and Makarov's~{\cite{Makarov1987}}
multiplies $3 \times 3$ matrices in $22$ scalar multiplications. These
formulas alone or as the base case of a Strassen scheme achieve what seems to
be the best known multiplication count for matrices of size up to $20$.

Exploiting these ideas leads to the following refinement of Theorem~\ref{thm:binsplit}. Similar results can be stated for general $\mathbbm{Z}$-algebras, using their rank and multiplicative complexity~{\cite{BuergisserClausenShokrollahi1997}}.

\begin{proposition} \label{prop:cste}
  Let $d'$ and $h'$ denote bounds on the
  degrees and heights (respectively) of $B(n)$ and $q(n)$ in
  Eq.~(\ref{eq:bsrec}).
  As $N, h' \rightarrow \infty$ ($s$ and $d'$ being fixed), the number
  of bit operations needed to compute
  the product tree $P (0, N)$ is at most $\bigl(C + o (1)\bigr) M \bigl(N (h' + d'
  \log N) \log (Nh')\bigr)$, with
  $C = (2 s^2 + 1) / 6$ in the FFT model, and
  $C = (3 \operatorname{MM} (s) + 1) / 4$ in the black-box model.
  Here $\operatorname{MM} (s)
  \leq (s^3 + 3 s^2) / 2$ is the algebraic complexity of $s \times
  s$ matrix multiplication over $\mathbbm{Z}$.
\end{proposition}

\begin{proof}
  Each node at the level $k$ (level~0 being the root)
  of the  tree essentially requires multiplying $s \times s$ matrices with
  entries in $\mathbbm{Z}[i]$ of height
  $H_k = N (h'+d'\log N)/2^{k+1}$,
  plus denominators of the same height. 
  In the FFT model, this may be done in $(2 s^2+1) M(H_k)$
  operations. 
  Since we assume $M (n) = \Omega (n \log n)$, we have
  $\sum_{k=0}^{\lceil \log N \rceil} 2^k M(H_k) \leq
  \frac12 M(H_0)$ (a remark attributed to D.~Stehlé
  in~{\cite{Zimmermann2006}}).
  Kramer's trick saves another factor $\frac32$.
  In the black-box model, the corresponding
  cost for one node is $(3\operatorname{MM}(s)+1) M(H_k)$ with
  Karatsuba's formula.
  Stehlé's argument applies again.
\end{proof}

Note that the previous techniques save time only for dense objects. In particular, one should not use the ``fast'' matrix multiplication formulae in the few bottom levels of product trees associated to recurrences of the form~(\ref{eq:bsrec}), since the matrices at the leaves are companion.

Continuing on this remark, these matrices often have some structure
that is preserved by successive multiplications. For instance, let $s_n =
\sum_{k = 0}^{n - 1} u_k$ where $(u_n)$ satisfies~(\ref{eq:rec}). It is easy
to compute a recurrence and initial conditions for $(s_n)$ and go on as above.
However, unrolling the recurrences (\ref{eq:rec}) and s$_{n + 1} - s_n = u_n $
simultaneously as
\begin{equation}\label{eq:recseries}
  \begin{pmatrix}
    u_{n + 1}\\
    \vdots\\
    u_{n + s - 1}\\
    u_{n + s}\\
    s_{n + 1}
  \end{pmatrix}
  = 
  \begin{pmatrix}
    & 1 &  &  & 0\\
    &  & \ddots &  & \vdots\\
    &  &  & 1 & 0\\
    \ast & \ast & \cdots & \ast & 0\\
    1 & 0 & \cdots & 0 & 1
  \end{pmatrix}
  \begin{pmatrix}
    u_n\\
    \vdots\\
    u_{n + s - 2}\\
    u_{n + s - 1}\\
    s_n
  \end{pmatrix}
\end{equation}
is more efficient. Indeed, all matrices in the product tree for the numerator
of~(\ref{eq:recseries}) then have a rightmost column of zeros, except for the
value in the lower right corner, which is precisely the denominator.
With the notation $\operatorname{MM}(s)$ of Proposition~\ref{prop:cste},
each product of these special matrices
uses $\operatorname{MM} (s) + s^2 + s + 1$ multiplications, {\emph{vs.}} $\operatorname{MM} (s
+ 1) + 1$ for the dense variant. Hence the formula~(\ref{eq:recseries}) is
more efficient as soon as $\operatorname{MM} (s + 1) - \operatorname{MM} (s) \geq s (s +
1)$, which is true both for the naïve multiplication algorithm and for
Waksman's algorithm (compare~{\cite{Zimmermann2006}}). In practice, on Ex.~\ref{ex:Gamma} below, if one puts $u_n = s_{n + 1} - s_n$ in
(\ref{eq:rec}, \ref{eq:bsrec}) instead of using~(\ref{eq:recseries}), the
computation time grows from 1.7~s to 2.7~s.

The same idea applies to any recurrence operator that can be factored. Further
examples of structure in product trees include even and odd \mbox{D-finite} series
({\emph{e.g.}}, {\cite[\S4.9.1]{BrentZimmermann2009-MCAv0.4}}). In all these
cases, the naïve matrix multiplication algorithm automatically benefits from
the special shape of the problem (because multiplications by constants have
negligible cost), while fast methods must take it into account explicitly.

\begin{remark}
A weakness of binary splitting is its comparatively large space complexity $\Omega (n \log n)$. Techniques to reduce it are known and used by efficient implementations in the case of recurrences of the first order~\cite{ChengHanrotThomeZimaZimmermann2007, GourdonSebah2001, CLN, Bellard2010}.
\end{remark}
\fixdefsectionspacing

\paragraph*{Implementation}
The implementation of binary splitting in NumGfun includes some of the tricks discussed in this section. FFT-based techniques are currently not used because they are not suited to implementation in the Maple language.
This implementation is exposed through two user-level functions,
\texttt{nth\_term} and \texttt{fnth\_term}, that allow to evaluate P-recursive sequences
(\texttt{fnth\_term} replaces the final gcd by a less
expensive division and returns a floating-point result). Additionally,
\texttt{gfun[rectoproc]}, which takes as input a recurrence and outputs a
procedure that evaluates its solution, automatically calls the
binary splitting code when relevant. Examples~\ref{ex:Motzkin}
and~\ref{ex:Gamma} illustrate the use of these functions on integer sequences
and convergent series respectively.

\begin{example}
  \label{ex:Gamma}{\cite[\S6.10]{Brent1978}} Repeated integration by parts of
  the integral representation of $\Gamma$ yields for $0 < \operatorname{Re} (z) < 1$
  \[ \Gamma (z) = \sum_{n = 0}^{\infty} \frac{e^{- t} t^{n + z}}{z (z + 1)
     \cdots (z + n)} + \int_t^{\infty} e^{- u} u^{z - 1} \mathrm{d} u. \]
  Taking $t = 29^3$, it follows that the sum $\sum_{n = 0}^{65000} u_n$,
  where $u_0 = 87 e^{- t}$ and $(3 n + 4) u_{n + 1} = 3 t u_n$, is within
  $10^{- 10000}$ of $\Gamma (1 / 3)$, whence $\Gamma (1 / 3) \simeq 2.67893
  \ldots ( \text{9990 digits}) \ldots 99978$. This computation takes~1.7~s.
\end{example}
\fixdefsectionspacing

\section{High-precision evaluation of \\  \mbox{D-finite}
functions} \label{sec:evaldiffeq}

We now recall the numerical evaluation algorithms used in NumGfun, and discuss their implementation.

Let $y (z) = \sum_n y_n z^n$ be a
\mbox{D-finite} series with radius of convergence~$\rho$ at the origin. Let $z \in
\mathbbm{Q}(i)$ be such that $\left| z \right| < \rho$ and $\operatorname{height} (z)
\leq h$. The sequence $(y_n z^n)$ is P-recursive, so that the binary
splitting algorithm yields a fast method for the high-precision evaluation of
$y (z)$. Here ``high-precision'' means that we let the precision required for
the result go to infinity in the complexity analysis of the algorithm.

More precisely, $(y_n z^n)$ is canceled by a recurrence relation of height
$O (h)$. By Theorem~\ref{thm:binsplit}, $y (z)$ may hence be computed to the
precision $10^{- p}$ in
\begin{equation}\label{eq:stepcomplexity}
  O \bigl(M \bigl(N (h + \log N)\bigr) \log(Nh)\bigr)
\end{equation}
bit operations, where $N$ is chosen so that $\left| y_{N ;} (z) \right|
\leq 10^{- p}$, {\emph{i.e.}} $N \sim p / \log (\rho /
\mathopen|z\mathclose|)$ if $\rho <
\infty$, and $N \sim p / (\tau \log p)$ for some $\tau$ if $\rho = \infty$.

In practice, the numbers of digits of (absolute) precision targeted in NumGfun
range from the hundreds to the millions. Accuracies of this order of magnitude
are useful in some applications to number
theory~{\cite{ChudnovskyChudnovsky1990}}, and in heuristic equality
tests~{\cite[\S5]{vdH2001}}. It can also happen that the easiest way to obtain
a moderate-precision output involves high-precision intermediate computations,
especially when the correctness of the result relies on pessimistic bounds.

\paragraph*{Analytic continuation}Solutions of the
differential equation~(\ref{eq:deq}) defined in the neighborhood of $0$ extend
by analytic continuation to the universal covering of $\mathbbm{C} \setminus
S$, where $S$ is the (finite) set of singularities of the equation. \mbox{D-finite}
functions may be evaluated fast at any point by a numerical version of the
analytic continuation process that builds on the previous algorithm~\cite{ChudnovskyChudnovsky1990}.
Rewrite~(\ref{eq:deq}) in matrix form
\begin{equation}\label{eq:matrixdeq}
  Y' (z) = A (z) Y (z), \hspace{2em} Y (z) = \Bigl( \frac{y^{(i)} (z)}{i!} \Bigr)_{0 \leq i < r} .
\end{equation}
This choice of $Y (z)$ induces, for all $z_0 \in \mathbbm{C} \setminus S$,
that of a family of {\emph{canonical solutions}} of~(\ref{eq:deq}) defined by
\[ y [z_0, j] (z) = (z - z_0)^j + O \bigl((z - z_0)^r\bigr), \hspace{2em} 0
   \leq j < r, \]
that form a basis of the solutions of~(\ref{eq:deq}) in the neighborhood of
$z_0$. Stated otherwise, the vector $\boldsymbol{y}[z_0] = (y [z_0, j])_{0
\leq j < r}$ of canonical solutions at $z_0$ is the first row
of the fundamental matrix $\boldsymbol{Y}[z_0] (z)$ of (\ref{eq:matrixdeq}) such
that $\boldsymbol{Y}[z_0] (z_0) = \operatorname{Id}$.

By linearity, for any path $z_0 \rightarrow z_1$ in $\mathbbm{C}
\setminus S$, there exists a {\emph{transition matrix}} $M_{z_0 \rightarrow
z_1}$ that ``transports initial conditions'' (and canonical solutions) along
the path:
\begin{equation}\label{eq:ancont}
  Y (z_1) = M_{z_0 \rightarrow z_1} Y (z_0), \hspace{1em} \boldsymbol{y}[z_0]
  (z) =\boldsymbol{y}[z_1] (z) M_{z_0 \rightarrow z_1} .
\end{equation}
This matrix is easy to write out explicitly:
\begin{equation}\label{eq:transition}
  M_{z_0 \rightarrow z_1} =\boldsymbol{Y}[z_0] (z_1) = \Bigl( \frac{1}{i!} y
  [z_0, j]^{(i)} (z_1) \Bigr)_{0 \leq i, j < r},
\end{equation}
evaluations at $z_1$ being understood to refer to the analytic continuation
path $z_0 \rightarrow z_1$. Transition matrices compose:
\begin{equation}\label{eq:composition}
  M_{z_0 \rightarrow z_1 \rightarrow z_2} = M_{z_1 \rightarrow z_2} M_{z_0
  \rightarrow z_1}, \hspace{1em} M_{z_1 \rightarrow z_0} = M_{z_0 \rightarrow
  z_1}^{- 1}.
\end{equation}

NumGfun provides functions to compute $M_{\gamma}$ for a given path $\gamma$
(\texttt{transition\_matrix}), and to evaluate the analytic continuation of
a \mbox{D-finite} function along a path starting at~$0$ (\texttt{evaldiffeq}). In
both cases, the path provided by the user is first subdivided into a new path
$z_0 \rightarrow z_1 \rightarrow \cdots \rightarrow z_m$, $z_{\ell} \in
\mathbbm{Q}(i)$, such that, for all $\ell$, the point $z_{\ell + 1}$ lies
within the disk of convergence of the Taylor expansions at $z_{\ell}$ of all
solutions of~(\ref{eq:deq}). Using~(\ref{eq:transition}), approximations
$\tilde{M}_{\ell} \in \mathbbm{Q}(i)^{r \times r}$ of $M_{z_{\ell} \rightarrow
z_{\ell + 1}}$ are computed by binary splitting.

More precisely, we compute all entries of~$\tilde M_\ell$ at once, as follows. For a generic
solution~$y$ of~\eqref{eq:deq}, the coefficients $u_n$ of $u (z) = y (z_{\ell} + z) = \sum_{n =
0}^{\infty} u_n z^n$ are canceled by a recurrence of order~$s$. Hence the
partial sums $u_{;n}(z)$ of~$u$ satisfy
\begin{equation} \label{eq:fastancont}
  \begin{pmatrix}
     U_{n + 1}\\
     u_{; n + 1} (z)
   \end{pmatrix} = 
   \underbrace{
   \begin{pmatrix}
     C (n) z & 0\\
     K & 1
   \end{pmatrix}
   }_{B(n)}
   \begin{pmatrix}
     U_n\\
     u_{; n} (z)
   \end{pmatrix},
\end{equation}
where $K = (1, 0, \ldots, 0)$ and $C (n) \in \mathbbm{Q}(n)^{s \times s}$.
Introducing an indeterminate~$\delta$, we let $z = z_{\ell + 1} - z_{\ell} + \delta \in \mathbbm{Q}(i) [\delta]$ and compute $B(N - 1) \ldots B(0) \mod \delta^r$ by binary splitting (an idea already used in~\cite{vdH1999}), for some suitable~$N$. The upper left blocks of the subproducts are \emph{kept factored} as a power of~$z$ times a matrix independent on~$z$. In other words, clearing denominators, the computation of each subproduct
\[
P = \frac{1}{d}
\begin{pmatrix}
D \cdot p & 0 \\
R & d
\end{pmatrix}
= P_{\mathrm{high}} P_{\mathrm{low}}
\qquad (p = \mathrm{numer}(z)^m)
\]
is performed as
$D \leftarrow D_{\mathrm{high}} D_{\mathrm{low}}$ ;
$p \leftarrow p_{\mathrm{high}} p_{\mathrm{low}}$ ;
$R \leftarrow p_{\mathrm{low}} (R_{\mathrm{high}} C_{\mathrm{low}}) +
     d_{\mathrm{high}} R_{\mathrm{low}}$ ;
$d \leftarrow d_{\mathrm{high}} d_{\mathrm{low}}$. 
The powers of~$z$ can be computed faster, but the saving is negligible.
Applying the row submatrix $R$ of the full product to the matrix $U_0 = (\frac{1}{i!} y[z_{\ell}, j]^{(i)}(z_\ell))_{0 \leqslant i < s, 0 \leqslant j < r}$ yields a row vector equal to
$(y [z_{\ell}, j]_{; N} (z_{\ell+1} + \delta))_{0 \leqslant j < r} + O (\delta^r)$,
each entry of which is a truncated power series whose coefficients are the entries of the corresponding column of $\tilde{M}_{\ell}$.
This way of computing $\tilde M_\ell$ is roughly $\min(r, s^{\omega-1})$ times more efficient than the fastest of the variants from~\cite{ChudnovskyChudnovsky1990, vdH1999}.

In the function \texttt{transition\_matrix}, we then form the product
$\tilde{M}_{m - 1} \cdots \tilde{M}_0$, again by binary splitting. In the case
of \texttt{evaldiffeq}, we compute only the first row
$\tilde{R}$ of $\tilde{M}_{m - 1}$, and we form the product $\tilde{R} 
\tilde{M}_{m - 2} \cdots \tilde{M}_0  \tilde{I}$, where $\tilde{I}$ is an
approximation with coefficients in $\mathbbm{Q}(i)$ of the vector $Y
(z_0)$ of initial conditions
(or this vector itself, if it has symbolic entries). The
whole computation is done on unsimplified fractions, controlling all
errors to guarantee the final result. Let us stress that no numerical instability occurs since all numerical operations are performed on rational numbers.
We postpone the discussion of approximation errors
(including the choice of~$N$)
to~\S\ref{sec:errorcontrol}.

\begin{example}
  Transition matrices corresponding to paths that ``go round'' exactly one
  singularity once are known as local monodromy matrices. A simple example is
  that of the equation $(1 + z^2) y'' + 2 zy' = 0$, whose solution space is
  generated by $1$ and $\arctan z$. Rather unsurprisingly, around~$i$:
\begin{verbatim}
> transition_matrix((1+z^2)*diff(y(z),z,y(z),z)
  +2*z*diff(y(z),z), y(z), [0,I+1,2*I,I-1,0], 20);
\end{verbatim}
\vspace*{-1.2ex}
\[ \begin{pmatrix}
     1.00000000000000000000 & 3.14159265358979323846\\
     0 & 1.00000000000000000000
\end{pmatrix} \]
  More generally, expressing them as entries of monodromy matrices is a way to
  compute many mathematical constants.
\end{example}

Another use of analytic continuation appears in Ex.~\ref{ex:Heun}: there,
despite the evaluation point lying within the disk of convergence, NumGfun
performs analytic continuation along the path $0
\rightarrow \frac{- 1}{2} \rightarrow \frac{- 3}{4} \rightarrow \frac{-
22}{25} \rightarrow \frac{- 47}{50} \rightarrow \frac{- 99}{100}$ 
to approach the singularity more easily
(\emph{cf.}~\cite[Prop.~3.3]{ChudnovskyChudnovsky1990}).

\paragraph*{The ``bit burst'' algorithm}The complexity
result~(\ref{eq:stepcomplexity}) is quasi-optimal for $h = O (\log p)$, but
becomes quadratic in~$p$ for $h = \Theta (p)$, which is the size of the
approximation $\tilde{z} \in \mathbbm{Q}(i)$ needed to compute $y (z)$ for
arbitrary $z \in \mathbbm{C}$. This issue can be avoided using analytic
continuation to approach~$z$ along a path $z_0 \rightarrow z_1 \rightarrow
\cdots \rightarrow z_m = \tilde{z}$ made of approximations of~$z$ whose
precision grow geometrically:
\begin{equation}\label{eq:bitburst}
  \mathopen| z_{\ell + 1} -
  z_{\ell} \mathclose| \leq 2^{- 2^{\ell}},
  \qquad
  \operatorname{height} (z_{\ell}) = O (2^{\ell}),
\end{equation}
thus balancing $h$ and $\left| z \right|$. This idea is due to Chudnovsky and
Chudnovsky~{\cite{ChudnovskyChudnovsky1990}}, 
who called it the bit burst algorithm,
and independently to van der Hoeven with a tighter complexity
analysis~{\cite{vdH1999,vdH2007c}}. The following proposition improves this
analysis by a factor $\log \log p$ in the case where $y$ has a finite radius
of convergence. No similar improvement seems to apply to entire functions.

\begin{proposition}
  Let $y$ be a \mbox{D-finite} function. One may compute a $2^{- p}$-approximation of
  $y$ at a point $z \in \mathbbm{Q}(i)$ of height $O(p)$ in $O (M (p \log^2
  p))$ bit operations.
\end{proposition}

\begin{proof}
  By~(\ref{eq:stepcomplexity}) and~(\ref{eq:bitburst}), the step $z_\ell
  \rightarrow z_{\ell + 1}$ of the bit-burst algorithm takes $O(M( p
  (2^\ell + \log p) \log p / 2^\ell))$ bit operations. Now
  \[ \sum^m_{\ell = 0} M \Bigl( \frac{p (2^\ell + \log p)}{2^\ell }
  \Bigr) \log p
     \leq M \Bigl( mn \log p + \sum_{\ell = 0}^{\infty} \frac{p \log^2
     p}{2^\ell} \Bigr) \]
  and $m = O (\log p)$, hence the total complexity.
\end{proof}

\begin{example}
  \label{ex:randeq}Consider the \mbox{D-finite} function $y$ defined by
  the following equation, picked at random:
  \begingroup \scriptsize
  \begin{equation*}
  \begin{split}
  \Bigl(\frac{5}{12} - \frac{1}{4} z + \frac{19}{24} z^2 - \frac{5}{24} z^3
  \Bigr) \, y^{(4)} + \Bigl(- \frac{7}{24} + \frac{2}{3} z + \frac{13}{24} z^2
  + \frac{1}{12} z^3 \Bigr)\, y''' \\ + \Bigr( \frac{7}{12} - \frac{19}{24} z
  + \frac{1}{8} z^2 + \frac{1}{3} z^3 \Bigr) \, y'' + \Bigl(-\frac{3}{4} +
  \frac{5}{12} z + \frac{5}{6} z^2 + \frac{1}{2} z^3 \Bigr)\, y' \\ + \Bigl(
  \frac{5}{24} + \frac{23}{24} z + \frac{7}{8} z^2 + \frac{1}{3} z^3\Bigr)\, y
  = 0, \\ y (0) = \frac{1}{24}, y' (0) = \frac{1}{12}, y'' (0) = \frac{5}{24},
  y''' (0) = \frac{5}{24}. 
  \end{split}
  \end{equation*}
  \endgroup
  The singular points are $z
  \approx 3.62$ and $z \approx 0.09 \pm 0.73 i$. By analytic continuation
  (along a path $0 \to \pi i$ homotopic to a segment) followed by bit-burst evaluation, we
  obtain
  \[ y (\pi i) \approx - 0.52299 \text{...(990 digits)...} 53279 - 1.50272
     \text{...} 90608 i \]
  after about 5~min.
  This example is among the ``most general'' that NumGfun can handle. Intermediate computations involve recurrences of order~$8$ and degree~$17$.
\end{example}
\fixdefsectionspacing

\section{Regular singular points}\label{sec:regsing}

The algorithms of the previous section extend to the case where the analytic
continuation path passes through {\emph{regular singular points}} of the
differential
equation~{\cite{ChudnovskyChudnovsky1986+1987,ChudnovskyChudnovsky1990,vdH2001}}.
Work is in progress to support this in NumGfun, with two main applications in
view, namely special functions (such as Bessel functions) defined using their
local behaviour at a regular singular point, and ``connection
coefficients'' arising in analytic
combinatorics~{\cite[\S\,VII.9]{FlajoletSedgewick2009}}. We limit ourselves to a sketchy (albeit
technical) discussion.

Recall that a finite singular point $z_0$ of a linear differential equation with
analytic coefficients is called {\emph{regular}} when all solutions $y (z)$
have moderate growth $y (z) = 1 / (z - z_0)^{O (1)}$ as $z \rightarrow
z_0$ in a sector with vertex at $z_0$, or equivalently when the equation
satisfies a formal property called Fuchs' criterion. The
solutions in the neighborhood of $z_0$ then have
a simple structure: for some $t \in \mathbbm{N}$, there exist linearly
independent formal solutions of the form (with $z_0=0$)
\begin{equation}\label{eq:Fuchs}
  y (z) = z^{\lambda}  \sum_{k = 0}^{t - 1} \phi_k (z) \log^k z,
  \hspace{1em} \lambda \in \mathbbm{C}, \hspace{1em} \phi_k \in
  \mathbbm{C}[[z]]
\end{equation}
in number equal to the order $r$ of the differential equation. Additionally,
the $\phi_k$ converge on the open disk centered at~0 extending to the
nearest singular point, so that the solutions~\eqref{eq:Fuchs}
in fact form a basis of analytic
solutions on any open sector with vertex at the origin contained in this disk.
(See for instance~{\cite{Poole1960}} for proofs and references.) 

Several extensions of the method of indeterminate coefficients used to obtain
power series solutions at ordinary points allow to determine the coefficients
of the series $\phi_k$. We proceed to revisit Poole's 
variant~{\cite[\S16]{Poole1960}} of Heffter's
method~{\cite[Kap.~8]{Heffter1894}} using the ``operator
algebra point of view'' on indeterminate coefficients: a recursive
formula for the coefficients of the series expansion of a solution is obtained
by applying simple rewriting rules to the equation. This formulation makes the algorithm simpler for our purposes (compare \cite{Tournier1987, ChudnovskyChudnovsky1986+1987, vdH2001}) and leads to a small complexity improvement. It also proves helpful in error control (§6).

Since our interest lies in the \mbox{D-finite} case, we assume that~$0$ is a regular
singular point of Equation~(\ref{eq:deq}). Letting $\theta = z
\frac{\mathrm{d}}{\mathrm{d} z}$, the equation rewrites as $L (z, \theta) \cdot y = 0$
where $L$ is a polynomial in two noncommutative indeterminates (and $L (z,
\theta)$ has no nontrivial left divisor in $K [z]$). Based
on~(\ref{eq:Fuchs}), we seek solutions as {\emph{formal logarithmic sums}}
\[ y (z) = \sum_{n \in \lambda +\mathbbm{Z}} \sum_{k \geq 0} y_{n, k} 
   \frac{\log^k z}{k!} z^n, \hspace{2em} \text{$\lambda \in \mathbbm{C}$} . \]
Let us call the double sequence $\boldsymbol{y}= (y_{n, k})_{n \in \lambda
+\mathbbm{Z}, k \in \mathbbm{N}}$ the {\emph{coefficient sequence}} of $y$.
The {\emph{shift operators}} $S_n$ and $S_k$ on such double sequences are
defined by $S_n \cdot \boldsymbol{y}= (y_{n + 1, k})$, and $S_k \cdot \boldsymbol{y}=
(y_{n, k + 1})$. The heart of the method lies in the following observation.

\begin{proposition}
  \label{prop:Poole}Let $y (z)$ be a formal logarithmic sum. The relation $L
  (z, \theta) \cdot y = 0$ holds (formally) i{f}f the coefficient sequence
  $\boldsymbol{y}$ of $y$ satisfies $L (S_n^{- 1}, n + S_k) \cdot \boldsymbol{y}=
  0$.
\end{proposition}

\begin{proof}
  The operators $z$ and $\theta$ act on logarithmic sums by $z \cdot y (z) =
  \sum_{n \in \lambda +\mathbbm{Z}} \sum_{k \geq 0} y_{n - 1, k} 
  \frac{\log^k z}{k!} z^n$ and  $\theta \cdot y (z) = \sum_{n \in \lambda
  +\mathbbm{Z}} \sum_{k \geq 0} (ny_{n, k} + y_{n, k + 1}) \frac{\log^k
  z}{k!} z^n$. Thus the coefficient sequence of $L (z, \theta) \cdot y$ is $L
  (S_n^{- 1}, n + S_k) \cdot \boldsymbol{y}$.
\end{proof}

Assume that $y (z)$ satisfies~(\ref{eq:deq}). Then $L (S_n^{- 1}, n + S_k)
\cdot \boldsymbol{y}= 0$; additionally, (\ref{eq:Fuchs}) implies that $y_{n, k}
= 0$ for $k \geq t$, which translates into $S_k^t \cdot \boldsymbol{y}= 0$.
Set $L (z, \theta) = Q_0(\theta) + R (z, \theta) z$, and fix $n \in \lambda
+\mathbbm{Z}$. If $Q_0 (n) \neq 0$, then $Q_0 (n + X)$ is invertible in $K
(\lambda) [[X]]$, and
\[ \boldsymbol{y} = - \bigl(Q_0 (n + S_k)^{- 1} \operatorname{mod}
S_k^t\bigr) R (S_n^{-
   1}, n + S_k) S_n^{-1} \cdot \boldsymbol{y}. \]
In general, let $\mu (n)$ be the multiplicity of $n$ as a root of~$Q_0$. Take
$T_n \in K (\lambda) [X]$ such that $T_n (X) \bigl(X^{- \mu (n)} Q_0 (n
+ X)\bigr) = 1 +
O (X^t)$ (explicitly, $T_n (S_k) = \sum_{v = 0}^{t - 1} \bigl(
\frac{\partial^v}{\partial X^v}  \frac{X^{\mu (n)}}{Q_0 (X)}\bigr)_{X = n} S_k^v$).
Then, it holds that
\begin{equation}\label{eq:regsingrec2}
  S_k^{\mu (n)} \cdot \boldsymbol{y} = - T_n (S_k) R (S_n^{- 1}, n +
  S_k) S_n^{-1}
  \cdot \boldsymbol{y}, 
\end{equation}
hence the $y_{n', k}$ with $n' < n$ determine $(y_{n, k})_{k \geq \mu
(n)}$, while the $y_{n, k}$, $k < \mu (n)$ remain free.

Coupled with the condition $y_{n, k} = 0$ for $n - \lambda < 0$ following from~(\ref{eq:Fuchs}), this implies that a
solution $y$ is characterized by $(y_{n, k})_{(n, k) \in E}$, where $E =\{(n,
k) \mathbin| k < \mu (n)\}$. As in \S\ref{sec:evaldiffeq}, this choice of initial
values induces that of canonical solutions (at $0$) $y [n, k]$ defined by
$y [n, k]_{n, k} = 1$ and $y [n, k]_{n', k'} = 0$ for  $(n', k') \in E
\setminus \{(n, k)\}$. The notion of transition matrix extends,
see~{\cite[\S4]{vdH2001}} for a detailed discussion.

Equation~(\ref{eq:regsingrec2}) is a ``recurrence relation'' that lends itself
to binary splitting. The main difference with the setting
of~\S\ref{sec:binsplit} is that the companion matrix of the ``recurrence
operator'' now contains truncated power series, {\emph{i.e.}}, $B \in K (n)
[[S_k]] / (S_k^t)$. The coefficients $y [u, v]_n = \sum_{k =
0}^{t - 1} y [u, v]_{n, k} \log^k z$ of canonical solutions may be computed
fast by forming the product ${B (n - 1) \cdots B (u)} \in K (\lambda) [[S_k]] /
(S_k^t)$ and applying it to the initial values $Y_u = (0, \ldots, 0, \log^v
z)^{\mathrm{T}}$. Compared to the algorithm of~{\cite{vdH2001}},
multiplications of power series truncated to the order $t$
replace multiplications of $t \times t$ submatrices, so that our method
is slightly faster. As in~§\ref{sec:evaldiffeq}, all entries of
$M_{z_0 \rightarrow z_1}$ may (and should) be computed at once.

\section{Error control}\label{sec:errorcontrol}

Performing numerical analytic continuation rigorously requires a number of
bounds that control the behaviour of the function under consideration. We now
describe how error control is done in NumGfun. Some of the ideas used in this section
appear scattered in~\cite{vdH1999}--\cite{vdH2007b}. NumGfun
relies almost exclusively on {\emph{a priori}} bounds;
see~{\cite[\S5.2]{MezzarobbaSalvy2009}} for pointers to alternative
approaches, and~{\cite{vdH2007b}} for further useful techniques, including how
to refine rough {\emph{a priori}} bounds into better {\emph{a posteriori}}
bounds.

We start by discussing the computation of {\emph{majorant series}} for the
canonical solutions of the differential equation. Then we describe how these
are used to determine, on the one hand, at which precision each transition
matrix should be computed for the final result to be correct, and on the other
hand, where to truncate the series expansions to
achieve this precision. Finally, we propose a way to limit the cost of
computing bounds in ``bit burst'' phases.

\begin{remark}
  In practice, numerical errors that happen while computing the error bounds themselves are not always controlled, due to limited support for interval arithmetic in Maple.
  However, we have taken some care to ensure that crucial steps rely on
  rigorous methods.
\end{remark}
\fixdefsectionspacing

\paragraph*{Majorant series}A formal power series $g \in \mathbbm{R}_+ [[z]]$
is a {\emph{majorant series}} of $y \in \mathbbm{C}[[z]]$, which we denote by
$y \trianglelefteq g$, if $\mathopen| y_n \mathclose| \leq g_n$ for all $n$.
If $y (z) \trianglelefteq g (z)$ and $\hat{y} (z) \trianglelefteq
\hat{g} (z)$, then
\begin{equation}\label{eq:majseriesprop} 
\begin{gathered}
  y (z) \leq g (\mathopen|z\mathclose|), \qquad
  \textstyle\frac{\mathrm d}{\mathrm d z}\,y(z) \trianglelefteq \textstyle \frac{\mathrm d}{\mathrm d z}\, g(z), \\
  y + \hat{y} \trianglelefteq g + \hat{g}, \qquad
  y \hat{y} \trianglelefteq g \hat{g}, \qquad
  y \circ \hat{y} \trianglelefteq g \circ \hat{g}.
\end{gathered}
\end{equation}
We shall allow majorant series to be formal logarithmic sums or matrices. The
relation $\trianglelefteq$ extends in a natural way: we write $\sum_{n, k} y_{n,
k} z^n_{} \log z^k \trianglelefteq \sum_{n, k} g_{n, k} z^n_{} \log z^k$
i{f}f $\mathopen| y_{n, k} \mathclose| \leq g_{n, k}$ for all $n$ and $k$, and  $Y
= (y_{i, j}) \trianglelefteq G = (g_{i, j})$ i{f}f $Y$ and $G$ have the
same format and $y_{i, j} \trianglelefteq g_{i, j}$ for all $i, j$. In
particular, for scalar matrices, $Y \trianglelefteq G$ if
$\mathopen| y_{i,
j} \mathclose| \leq g_{i, j}$ for all $i, j$. The
inequalities~(\ref{eq:majseriesprop}) still hold. 

\paragraph*{Bounds on canonical solutions}The core of the error control is a
function that computes $g (z)$ such that
\begin{equation}
  \label{eq:majcansol}
  \forall j, \hspace{1em} y[z_0,j](z_0+z) \trianglelefteq g (z)
\end{equation}
(in the notations of~\S\ref{sec:evaldiffeq}) given Equation~(\ref{eq:deq}) and
a point $z_0$. This function is called at each step of analytic continuation.

The algorithm is based on that of~{\cite{MezzarobbaSalvy2009}} (``SB'' in what
follows), designed to compute ``asymptotically tight'' {\emph{symbolic}}
bounds. Run over
$\mathbbm{Q}(i)$ instead of $\mathbbm{Q}$, SB returns (in the case of convergent series) a power series satisfying~(\ref{eq:majcansol}) of the form
$g (z) = \sum_{n = 0}^{\infty} n!^{- \tau} \alpha^n \phi (n) z^n$, where
$\tau \in \mathbbm{Q}_+$, $\alpha \in \bar{\mathbbm{Q}}$, and $\phi (n) =
e^{o (n)}$. The series~$g$ is specified by $\tau$, $\alpha$, and other
parameters of no interest here that define~$\phi$. The tightness property
means that $\tau$ and $\alpha$ are the best possible.

However, the setting of numerical evaluation differs from that of symbolic
bounds in several respects: (i)~the issue is no more to obtain human-readable
formulae, but bounds that are easy to evaluate; (ii)~bound computations should
be fast; (iii)~at regular singular points, the fundamental solutions are
logarithmic sums, not power series. For space reasons, we only sketch the differences between SB and the variant we use for error control.

First, we take advantage of~(i) to solve~(ii). The most important change is that the parameter~$\alpha$ is replaced by an approximation $\tilde\alpha\geq\alpha$. This avoids computations with algebraic numbers, the bottleneck of SB. Strictly speaking, it also means that we are giving up the tightness of the bound. However, in constrast with the ``plain'' majorant series
method~{\cite{vdH2001,vdH2003,vdH2007b}}, we are free to take $\tilde{\alpha}$
arbitrarily close to $\alpha$, since we do not use the ratio $(\alpha /
\tilde{\alpha})^n$ to mask subexponential factors.
The algorithms from~\cite{MezzarobbaSalvy2009} adapt without difficulty. Specifically, Algorithms~3 and 4 are modified to work with~$\tilde\alpha$ instead of~$\alpha$. Equality tests between ``dominant roots'' (Line~9 of Algorithm~3, Line~5 of Algorithm~4) can now be done numerically and are hence much less expensive. As a consequence, the parameter~$T$ from~\S3.3 is also replaced by an upper bound. This affects only the parameter $\phi(n)$ of the bound, which is already pessimistic.

The issue~(iii) can be dealt with too. One may compute a series $g$ such that $y
(z) \trianglelefteq g (z) \sum_{k = 0}^{t - 1} \frac{\log^k z}{k!}$ 
(with the notations of~\S\ref{sec:regsing})
by
modifying the choice of~$K$ in~{\cite[\S3.3]{MezzarobbaSalvy2009}} so
that
$\sum_{k = 0}^{t - 1} \bigl| \bigl( \frac{\partial^k}{\partial X^k}  \frac{X^{\mu
(n)}}{Q_0 (X)}\bigr)_{X = n} \bigr| \leq K / n^r$, and
replacing~{\cite[Eq.~(14)]{MezzarobbaSalvy2009}} by Eq.~(\ref{eq:regsingrec2})
above in the reasoning.

\paragraph*{Global error control}We now consider the issue of controlling how
the approximations at each step of analytic continuation accumulate. Starting
with a user-specified error bound $\epsilon$, we first set $\epsilon'$
so that an $\epsilon'$-approximation $\tilde{r} \in \mathbbm{Q}(i)$ of the
exact result $r$ rounds to $\diamond ( \tilde{r}) \in \bigcup_{n \in \mathbbm{N}} 10^{-n} \mathbbm{Z}[i]$ with $\mathopen| \diamond ( \tilde{r}) - r\mathclose| < \epsilon$. No other
{\emph{rounding}} error occur, since the rest of the computation is done on
objects with coefficients in $\mathbbm{Q}(i)$. However, we have to choose the
precision at which to compute each factor of the product $\Pi = \tilde{R} 
\tilde{M}_{m - 1} \cdots \tilde{M}_0 \tilde{I}$ (in \texttt{evaldiffeq}, and
of similar products for other analytic continuation functions) so that
$\mathopen|
\tilde{r} - r\mathclose| < \epsilon'$.

As is usual for this kind of applications, we use the Frobenius norm, defined
for any (not necessarily square) matrix by $\mathopen\| (a_{i, j})
\mathclose\|_{\mathrm{F}} = ( \sum_{i, j} \left| a_{i, j} \right|^2)^{1 / 2}$. The
Frobenius norm satisfies $\mathopen\| AB \mathclose\|_{\mathrm{F}} \leq \mathopen\| A
\mathclose\|_{\mathrm{F}}  \mathopen\| B \mathclose\|_{\mathrm{F}}$ for any matrices $A,
B$ of compatible dimensions. If $A$ is a square $r \times r$ matrix,
\begin{equation}\label{eq:propFrobenius}
  \mathopen\| A \mathclose\|_{\infty} \leq \mathopen\interleave A
  \mathclose\interleave_2 \leq
  \mathopen\| A \mathclose\|_{\mathrm{F}} \leq r \mathopen\| A \mathclose\|_{\infty}
\end{equation}
where $\mathopen\interleave \mathord\cdot \mathclose\interleave_2$ is the matrix norm induced by the
Euclidean norm while $\mathopen\| \mathord\cdot \mathclose\|_{\infty}$ denotes the
{\emph{entrywise}} uniform norm. Most importantly, the computation of $\mathopen\|
A \mathclose\|_{\mathrm{F}}$ is numerically stable, and if $A$ has entries in
$\mathbbm{Q}(i)$, it is easy to compute a good upper bound on $\mathopen\| A
\mathclose\|_{\mathrm{F}}$ in floating-point arithmetic.

We bound the total error $\epsilon_{\Pi}$ on $\Pi$ using repeatedly the
rule $\| \tilde{A}  \tilde{B} - AB\|_{\mathrm{F}} \leq \| \tilde{A}
\|_{\mathrm{F}} \| \tilde{B} - B\|_{\mathrm{F}} +\| \tilde{A} -
A\|_{\mathrm{F}} \|B\|_{\mathrm{F}}$. From there, we compute precisions
$\epsilon_A$ such that having $\| \tilde{A} - A\|_{\mathrm{F}} <
\epsilon_A$ for each factor $A$ of $\Pi$ ensures that $\epsilon_{\Pi} <
\epsilon'$. Upper bounds on the norms $\| \tilde{A} \|_{\mathrm{F}} $ and
$\|A\|_{\mathrm{F}}$ appearing in the error estimate are computed either from
an approximate value of $A$ (usually $\tilde{A}$ itself) if one is available
at the time, or from a matrix~$G$ such that $A \trianglelefteq G$ given
by~(\ref{eq:majseriesprop}, \ref{eq:majcansol}).

\paragraph*{Local error control}
Let us turn to the computation of individual
transition matrices. We restrict to the case of ordinary points. Given a precision $\epsilon$ determined by
the ``global error control'' step, our task is to choose~$N$ such that $\|
\tilde{M} - M_{z_0 \rightarrow z_1} \|_{\mathrm{F}} \leq \epsilon$ if
$\tilde{M}$ is computed by truncating the entries of~(\ref{eq:transition}) to
the order~$N$. If $g$ satisfies~(\ref{eq:majcansol}), it suffices that $g_{N
;}^{(i)} (|z_1 - z_0 |) \leq \frac{i!}{r} \epsilon$ for all~$i$, as
can be seen from~(\ref{eq:majseriesprop}, \ref{eq:propFrobenius}). We find a
suitable~$N$ by dichotomic search. For the family of majorant series~$g$ used in
NumGfun, the $g_{N ;}^{(i)} (x)$ are not always easy to evaluate, so we bound
them by expressions involving only elementary functions~{\cite[\S4.2]{MezzarobbaSalvy2009}}. The basic idea, related to the
saddle-point method from asymptotic analysis, is that if $x \geq 0$,
inequalities like $g_{n ;} (x) \leq (x / t_n)^n g (t_n) (1 - x / t_n)^{-
1}$  are asymptotically tight for suitably chosen $t_n$.

\paragraph*{Points of large bit-size}Computing the
majorants~(\ref{eq:majcansol}) is expensive when the point $z_0$ has large
height. This can be fixed by working with an approximate value $c$ of $z_0$ to
obtain a bound valid in a neighborhood of $c$ that contains~$z_0$. This
technique is useful mainly in ``bit burst'' phases (where, additionally, we
can reuse the same $c$ at each step).

Assume that $\boldsymbol{Y}[c] (c + z) \trianglelefteq G (z)$ for some $c
\approx z_0$. Since $\boldsymbol{Y}[z_0] (z_0 + z) =\boldsymbol{Y}[c] (c + (z_0 -
c) + z) M_{c \rightarrow z_0}^{- 1}$ by~(\ref{eq:ancont}), it follows that
$\boldsymbol{Y}[z_0] (z_0 + z) \trianglelefteq G (|z_0 - c| +
z)\,\mathopen\|M_{c
\rightarrow z_0}^{- 1} \mathclose{\|} \boldsymbol{1}$ where~$\boldsymbol{1}$ is a square matrix
filled with ones. Now $M_{c \rightarrow c + z} = \operatorname{Id} + O
(z)$, whence $M_{c \rightarrow c + z} - \operatorname{Id} \trianglelefteq G (z) -
G (0)$. For $|z_0 - c| < \eta$, this implies $\|M_{c \rightarrow z_0} -
\operatorname{Id} \|_{\mathrm{F}} \leq \delta := \|G (\eta) - G
(0)\|_{\mathrm{F}}$. Choosing $\eta$ small enough that $\delta < 1$,
we get
$\|M_{c \rightarrow z_0}^{- 1} \|_{\mathrm{F}} \leq (1 - \delta)^{- 1}$.
If $G$ was computed from~(\ref{eq:majcansol}), {\emph{i.e.}},
for $G (z) = \bigl( \frac{1}{j!} g^{(j)} (z)\bigr)_{i, j}$,
this finally gives the bound 
\[ y [z_0, j]
(z_0+z) \trianglelefteq \frac{r}{j! \, (1 - \delta)} g^{(j)} (\eta + z), \] valid
for all $z_0$ in the disk $|z_0 - c| < \eta$.

Note however that simple differential equations have solutions like $\exp (K
/ (1 - z)$) with large $K$. The condition $\delta < 1$ then forces us to take
$c$ of size $\Omega (K)$. Our strategy to deal with this issue in NumGfun is
to estimate $K$ using a point~$c$ of size $O (1)$ and then to choose a more
precise $c$ (as a last resort, $c = z_0$) based on this value if necessary.

\begin{remark}
  If the evaluation point $z$ is given as a program, a similar
  reasoning allows to choose automatically an approximation
  $\tilde{z} \in \mathbbm{Q}(i)$ of $z$ such that $|y ( \tilde{z}) - y (z) |$
  is below a given error bound~{\cite[\S4.3]{vdH1999}}. In other applications,
  it is useful to have bounds on transition matrices that hold uniformly for
  {\emph{all}} small enough steps in a given domain. Such bounds may be
  computed from a majorant differential equation with constant
  coefficients~{\cite[\S5]{vdH2007}}.
\end{remark}
\fixdefsectionspacing

\section{Repeated evaluations}\label{sec:diffeqtoproc}

Drawing plots or computing integrals numerically requires to evaluate the same
function at many points, often to a few digits of precision only. NumGfun
provides limited support for this through the function
\texttt{diffeqtoproc}, which takes as input a \mbox{D-finite} function $y$, a
target precision $\epsilon$, and a list of disks, and returns a Maple
procedure that performs the numerical evaluation of $y$. For each disk $D
=\{z, \mathopen|z - c \mathclose| < \rho\}$, \texttt{diffeqtoproc} computes a polynomial
$p \in \mathbbm{Q}(i) [z]$ such that $\mathopen| \diamond (p (z)) - y
(z) \mathclose|
< \epsilon$ for $z \in D \cap \mathbbm{Q}(i)$, where $\diamond$ again denotes rounding to complex decimal. The procedure returned by \texttt{diffeqtoproc} uses the precomputed $p$ when possible, and
falls back on calling \texttt{evaldiffeq} otherwise.

The approximation polynomial $p$ is constructed as a linear combination of truncated Taylor series of fundamental solutions $y[c, j]$, with coefficients obtained by
numerical analytic continuation from $0$ to $c$.
The way we choose expansion
orders is
similar to the error control techniques of~§\ref{sec:errorcontrol}: we first compute a bound $B_{\operatorname{can}}$ on the
fundamental solutions and their first derivatives on the disk~$D$. The vector $Y (c)$ of ``initial values'' at
$c$ is computed to the precision $\epsilon' / B_{\operatorname{can}}$ where
$\epsilon' \leq \epsilon / (2 r)$. We also compute
$B_{\operatorname{ini}} \geq \|Y (c)\|_{\mathrm{F}}$. Each $y [c, j]$ is expanded
to an order $N$ such that $\mathopen\| y [c, j]_{N ;} \mathclose\|_{\infty, D}
\leq \epsilon' / B_{\operatorname{ini}}$, so that finally $|p (z) - y (z) |
\leq \epsilon$ for $z \in D$.

The most important feature of \texttt{diffeqtoproc} is that it produces certified results. At low precisions and in the absence of singularities, we expect that interval-based numerical solvers will perform better while still providing (\emph{a posteriori}) guarantees. Also note that our simple choice of~$p$ is far from optimal. If approximations of smaller degree or height are required, a natural approach is to aim for a slightly smaller error $\mathopen{\|}y - p\mathclose{\|}_{\infty, D}$ above, and then replace~$p$ by a polynomial~$\tilde{p}$ for which we
can bound $\mathopen{\|}p - \tilde{p} \mathclose{\|}_{\infty}$~{\cite[\S6.2]{vdH2007b}}.

\begin{example}
  The following plot of the function~$y$ defined by  
  $(z-1)\,y''' - z(2z-5)\,y'' - (4z-6)\,y' + z^2(z-1)\,y$ with the initial values
  $y(0)=2$, $y'(0)=1$, $y''(0)=0$ was obtained using polynomial approximations on several disks that cover the domain of the plot while avoiding the pole $z=1$.   The whole computation takes about 9~s.
  Simple numerical integrators typically fail to evaluate~$y$ beyond~$z=1$.
  \begin{center}
  \includegraphics{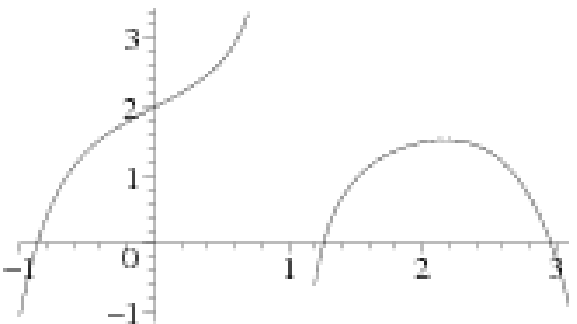}
  \end{center}
\end{example}
\fixdefsectionspacing

\section{Final remarks}

Not all of NumGfun was described in this article. The symbolic bounds
mentioned in~\S\ref{sec:errorcontrol} are also implemented, with functions
that compute majorant series or other kinds of bounds on rational functions
(\texttt{bound\_ratpoly}), \mbox{D-finite} functions (\texttt{bound\_diffeq} and
\texttt{bound\_diffeq\_tail}) and P-recursive sequences
(\texttt{bound\_rec} and \texttt{bound\_rec\_tail}). This implementation
was already presented in~{\cite{MezzarobbaSalvy2009}}.

Current work focuses on adding support for evaluation ``at regular singular points'' (as outlined in \S\ref{sec:regsing}), and improving performance. The development version of NumGfun already contains a second implementation of binary splitting, written in~C and called from the Maple code. In the longer term, I~plan to rewrite other parts of the package ``from the bottom up'', both for efficiency reasons and to make useful subroutines independant of Maple.

\paragraph*{Acknowledgements}I am grateful to my advisor, B.~Salvy, for encouraging me to conduct this work and offering many useful comments. Thanks also to A.~Benoit, F.~Chyzak, P.~Giorgi, J.~van der Hoeven and A.~Vaugon for stimulating conversations, bug reports, and/or comments on drafts of this article, and to the anonymous referees for helping make this paper more readable.

This research was supported in part by the MSR-INRIA joint research center.

\end{document}